\documentclass[conference, 10pt]{IEEEtran}

\usepackage{setspace}
\usepackage{amsmath}
\usepackage{amssymb}
\usepackage{mathrsfs}
\usepackage{amsthm}
\usepackage{multirow}
\usepackage{color}
\usepackage{array}
\usepackage{url}
\usepackage{comment}
\usepackage{enumerate}
\usepackage{eucal}

\usepackage{algorithm}
\usepackage[noend]{algorithmic}
\usepackage{cite}
\usepackage[table]{xcolor}
\usepackage{blkarray}

\usepackage{stmaryrd}

\usepackage[top=0.6in, bottom=0.6in, left=0.51in, right=0.51in]{geometry}

\usepackage{tikz}
\usetikzlibrary{arrows, patterns, shapes.arrows, decorations.pathmorphing}

\IEEEoverridecommandlockouts

\theoremstyle{plain}
\newtheorem{theorem}{Theorem}
\newtheorem{corollary}[theorem]{Corollary}
\newtheorem{lemma}[theorem]{Lemma}
\newtheorem{proposition}[theorem]{Proposition}

\theoremstyle{definition}
\newtheorem{definition}{Definition}
\newtheorem{example}[definition]{Example}

\newtheorem{remark}[definition]{Remark}

\newcommand{\sS}{{\sf S}}

\newcommand{\onezero}{{\underline{10}\,}}
\newcommand{\zeroone}{{01\,}}
\newcommand{\ins}[1]{{\color{blue}#1}}

\DeclareMathAlphabet{\mathbfsl}{OT1}{ppl}{b}{it} 

\newcommand{\ve}{\mathbfsl{e}}

\newcommand{\va}{\mathbfsl{a}}

\newcommand{\vf}{\mathbfsl{f}}

\newcommand{\vu}{\mathbfsl{u}}

\newcommand{\vx}{\mathbfsl{x}}
\newcommand{\vy}{\mathbfsl{y}}
\newcommand{\vz}{\mathbfsl{z}}

\newcommand{\vA}{\mathbfsl{A}}
\newcommand{\vB}{\mathbfsl{B}}
\newcommand{\vC}{\mathbfsl{C}}

\newcommand{\vU}{\mathbfsl{U}}
\newcommand{\vV}{\mathbfsl{V}}
\newcommand{\vX}{\mathbfsl{X}}

\newcommand{\vY}{\mathbfsl{Y}}
\newcommand{\vZ}{\mathbfsl{Z}}

\newcommand{\cC}{\mathcal{C}}
\newcommand{\cD}{\mathcal{D}}
\newcommand{\cE}{\mathcal{E}}
\newcommand{\cF}{\mathcal{F}}

\newcommand{\cX}{\mathcal{X}}

\newcommand{\dec}{{\tt DEC}}
\newcommand{\fail}{{\tt FAILURE}}
\newcommand{\etal}{{\em et al.}}

\newcommand{\todo}[1]{{\color{red} TODO: #1}}

\title{Sequence Reconstruction Problem for Deletion Channels: A Complete Asymptotic Solution}

\author{
 \IEEEauthorblockN{
 	Van Long Phuoc Pham,
 	Keshav Goyal, and
	Han Mao Kiah 
	}
	\IEEEauthorblockA{School of Physical and Mathematical Sciences, Nanyang Technological University, Singapore} 
 \IEEEauthorblockA{Emails: 
 	phuoc002, keshav002, hmkiah@ntu.edu.sg
 	\vspace{-5mm}}
 	}

\begin{document}
\date{}

\maketitle

\hspace*{-15pt}
\begin{abstract}
Transmit a codeword $\vx$, that belongs to 
an $(\ell-1)$-deletion-correcting code of length $n$, over a $t$-deletion channel for some $1\le \ell\le t<n$.
Levenshtein, in 2001, proposed the problem of determining $N(n,\ell,t)+1$, the minimum number of distinct channel outputs required to uniquely reconstruct $\vx$. 
Prior to this work, $N(n,\ell,t)$ is known only when $\ell\in\{1,2\}$.
Here, we provide an asymptotically exact solution for all values of $\ell$ and $t$. 
Specifically, we show that $N(n,\ell,t)=\binom{2\ell}{\ell}/(t-\ell)! n^{t-\ell} - O(n^{t-\ell-1})$ and 
in the special instance where $\ell=t$, we show that $N(n,\ell,\ell)=\binom{2\ell}{\ell}$.
We also provide a conjecture on the exact value of $N(n,\ell,t)$ for all values of $n$, $\ell$, and $t$.
\end{abstract}

\section{Introduction}

The {\em sequence reconstruction problem} \cite{Levenshtein.2001}, introduced by Levenshtein in 2001, considers a communication scenario where the sender transmits a codeword $\vx$ from some codebook $\cC$ over a number of noisy channels. 
The receiver then obtains all the noisy channel outputs and attempts to reconstruct the transmitted codeword $\vx$. 
So, for a fixed codebook $\cC$ and channel, our task is to determine the minimum number of channels that are required for unique reconstruction.
While the sequence reconstruction 
problem was first motivated by applications in biology and chemistry, the problem has received renewed interest because of certain emerging data storage media. These modern storage media relies on technologies that provide users with multiple cheap, noisy reads and examples include DNA-based data storage \cite{Church.etal:2012,Goldman.etal:2013,Yazdi.etal:2015b,Organick2017,Lenz.2019} and racetrack memories \cite{Parkin.2008,Chee.2018}.

In this work, we focus on channels that introduce {\em deletions}. 
Formally, when a word of length $n$ is sent through a $t$-deletion channel, a subsequence of length $n-t$ is received. 
A $t$-deletion correcting code $\cC$ is then a subset of length-$n$ binary words such that for any codeword $\vx \in \cC$, we are able to uniquely identify $\vx$ from any length-$(n-t)$ subsequence of $\vx$.
In his seminal work \cite{Levenshtein.2001,Levenshtein.2001.jcta}, Levenshtein studied the sequence reconstruction problem for the $t$-deletion channel. 
For the case where $\cC$ is the set of all binary sequences, Levenshtein determined the minimum number of channel outputs required for unique reconstruction. 
However, when $\cC$ is a $\ell$-deletion-correcting code where $\ell\ge 1$, little results are known.
Only recently, Gabrys and Yaakobi \cite{Gabrys.2018} solved the sequence reconstruction problem for the $t$-deletion channel when $\cC$ is a single-deletion-correcting codes and 
in the paper, they noted that the problem remains open for the case where $\cC$ is an $\ell$-deletion correcting codes for $\ell\ge 2$.

However, little progress was made on this open problem. Nevertheless, there was a slew of related results.
The sequence reconstruction problem was solved in the following instances:
\begin{itemize}
	\item when the channel involves insertion errors only and the codebook is any $e$-insertion error-correcting
	code \cite{Sala.2017};
	\item when the channel involves combinations of single substitution and
	single insertion error and the codebook comprises of all binary words \cite{Abu-Sini.2021}. 
\end{itemize}
Recently, in \cite{Cai.2021.recon,Chrisnata.2020,Chrisnata.2021}, the authors study the problem of {\em code design} under such scenarios.
Specifically, they fix the number of channel outputs available to the receiver and
design codes that allow the receiver to uniquely reconstruct the transmitted codeword.

In this work, we revisit the open problem posed by Levenshtein~\cite{Levenshtein.2001}(and later by Gabrys and Yaakobi~\cite{Gabrys.2018}) and provide an {\em asymptotic solution for all values of $\ell\ge 2$}. 
Specifically, let $1 \le \ell \le t$ and we transmit codewords from an $(\ell-1)$-deletion-correcting code over a $t$-deletion channel. 
In this work, we first show that the number of channels required for unique reconstruction is upper bounded by the quantity $\binom{2\ell}{\ell}n^{t-\ell}/(t-\ell)!$ (see Theorem~\ref{thm:upper}).
Subsequently, we provide a matching lower bound.
That is, we construct a pair of codewords $\vX$ and $\vY$ with Levenshtein distance at least $\ell$ and show that the number of channels required to disambiguate $\vX$ or $\vY$ is at least $\binom{2\ell}{\ell}n^{t-\ell}/(t-\ell)! - O(n^{t-\ell-1})$. 
This therefore implies that our estimate is asymptotically exact (see Theorem~\ref{thm-main}).
Furthermore, in the special case where $\ell=t$, we determine that $N(n,\ell,t)=\binom{2\ell}{\ell}$.

Before we formally state our contributions, we briefly remark on our proof techniques.
The main difficulty lies with the proof of the upper bound. 
While our arguments bear certain similarities to that of Levenshtein \cite{Levenshtein.2001} and Gabrys and Yaakobi \cite{Gabrys.2018}, an analysis that mimics these works is too tedious.
Instead, we turn to techniques in {\em subsequence combinatorics}.
Of particular interest is~\cite{Elzinga.2008}, where Elzinga \etal{} developed certain recursion rules and used dynamic programming to provide quadratic-time algorithms to enumerate certain subsequence problems. In our paper, we modify these recursion rules to provide an inductive proof of the upper bound in Section~\ref{sec:upper}.
We believe that the recursion rules developed in this work will provide insights for other problems related to sequence reconstruction.


\section{Preliminaries}

Let $\Sigma$ denote the binary alphabet $\{0,1\}$. 
We use $\Sigma^n$ to denote the set of all length-$n$ binary sequences. 

Let $\vx\in \Sigma^n$. Then the {\em deletion ball of radius $t$ centered at $\vx$} is defined to be the set of all length-$(n-t)$ subsequences of $\vx$ and this ball is denoted by $\cD_t(\vx)$.
Given two binary sequences $\vx$ and $\vy$ with $|\vy|=|\vx|+k$, we are interested in the intersection of their deletion balls and 
we use $\cD(\vx,\vy,t,t+k)$ to denote the set $\cD_t(\vx)\cap\cD_{t+k}(\vy)$.
Furthermore, we define their {\em Levenshtein distance} to be $d_L(\vx,\vy)\triangleq\min \left\{ t\ge 0: \cD\left(\vx, \vy, t, t+k\right) \neq \varnothing \right\}$.
Equivalently, if $d_L(\vx,\vy)=\ell$, we have that 
\[
	\cD(\vx, \vy, \ell, \ell+k) \neq \varnothing \text{ and } \cD(\vx, \vy, \ell-1, \ell+k-1) = \varnothing\,. 
\]
Hence, a codebook $\cC$ is an {\em $(\ell-1)$-deletion-correcting code} if $d_L(\vx,\vy)\ge \ell$ for all distinct $\vx,\vy\in \cC$. 

\vspace{2mm}

We now formally define our problem statement. 
For $1\le \ell \le t< n$, the task of {\em sequence reconstruction problem} for deletion channels is to determine the following quantity.
\vspace{-3mm}

{\small 
\begin{equation}\label{eq:Ntell}
	N(n, \ell, t ) \triangleq \max \{ |\cD(\vx,\vy,t,t+k) | :\vx,\vy\in \Sigma^n, d_L(\vx,\vy)\ge \ell \}.
\end{equation}
}%
Suppose we have an $(\ell-1)$-deletion correcting code $\cC$ of length $n$.
If a codeword from $\cC$ is transmitted over a $t$-deletion channel, 
Levenshtein showed that $N(n,\ell,t)+1$ distinct channel outputs are sufficient to allow unique reconstruction of the transmitted word \cite{Levenshtein.2001}.

For $\ell\in\{1,2\}$, the exact values $N(n,\ell,t)$ have been determined in \cite{Levenshtein.2001, Gabrys.2018}.
To state these results, we require the maximum size of a $t$-deletion ball.
Specifically, for $0\le  t < n$, we use $D(n,t)$ to denote the quantity $\max \{|\cD_t(\vx)| : \vx\in \Sigma^n \}$. We know from \cite{Calabi.1967} that 
\begin{equation}\label{eq:l0}
	D(n,t) = \sum_{i=0}^t \binom{n-t}{i} = \frac{1}{t!} n^t - O(n^{t-1})\, ,
\end{equation}
\noindent and the maximum is achieved when the $\vx$ is alternating.
For convenience, we extend the domain of \eqref{eq:Ntell} to include $\ell=0$ and hence, $N(n,0,t)$ is given by $D(n,t)$.

When $\ell=1$, we have the following landmark result of Levenshtein.

\begin{theorem}[Levenshtein \cite{Levenshtein.2001}]\label{thm:l1}
	For $1\le t< n $,
	\begin{equation}\label{eq:l1}
		N(n,1,t) = 2D(n-2,t-1) = \frac{2}{(t-1)!}n^{t-1} - O(n^{t-2}).
	\end{equation}
\end{theorem}

The result for $\ell=2$ is obtained more than a decade later by Gabrys and Yaakobi.

\begin{theorem}[Gabrys and Yaakobi \cite{Gabrys.2018}]\label{thm:l2}
	For $2\le t< n$ and $n\ge 8$,
{\small
\begin{align}
	N(n,2,t) & = 2D(n-4,t-2) + 2D(n-5,t-2) + 2D(n-7,t-2) \notag \\
	& \hspace*{2.5cm}+ D(n-6,t-3) + D(n-7,t-3) \notag \\
	& = \frac{6}{(t-2)!}n^{t-2} - O(n^{t-3}). \label{eq:l2}
\end{align}
}
\end{theorem}

\subsection{Our Contributions}

In this paper, we provide asymptotically exact estimates of $N(n,\ell,t)$ for {\em all} values of $0\le \ell \le t$ . Specifically, we establish the following theorem.

\begin{theorem}[Main Theorem]\label{thm-main}
	For $0\le \ell \le t <  n $, we have that
	\begin{equation}\label{eq:lmain}
		N(n,\ell,t) = \frac{\binom{2\ell}{\ell}}{(t-\ell)!} n^{t-\ell} - O(n^{t-\ell-1})\,.
	\end{equation}
\end{theorem}

\begin{remark}\hfill
	\begin{itemize}
		\item Observe that when $\ell\in\{0,1,2\}$, the main theorem recovers the asymptotic estimates of \eqref{eq:l0}, \eqref{eq:l1}, and \eqref{eq:l2}, respectively.
		\item	In \eqref{eq:l0}, \eqref{eq:l1}, \eqref{eq:l2}, \eqref{eq:lmain}, we use the usual big $O$ notation where the asymptotics is measured in terms of $n$. That is, $f(n)=O(g(n))$ means that $\limsup_{n\to\infty} f(n)/g(n)$ is bounded by some constant $C$. 
		Also, we emphasize that whenever the big $O$ notation is used,  both $f(n)$ and $g(n)$ are positive functions. Hence, for example, from \eqref{eq:l0}, we have that $D(n,t)$ is at most $n^t/t!$\,.
	\end{itemize}
\end{remark}

Here, we outline the proof for the main theorem.
First, we demonstrate an upper bound for $N(n,\ell,t)$ in Section~\ref{sec:upper}. 
In particular, we study a general version of the quantity $N(n,\ell,t)$ where the transmitted sequences are of different lengths. 
Specifically, we consider two binary sequences $\vx,\vy$ with $|\vy|=|\vx|+k$ and $k\ge 0$, set $\Delta(\vx,\vy , t, t+k)$ to be the size of the intersection $\cD(\vx,\vy , t, t+k)$. 
Section~\ref{sec:upper} is dedicated to an induction proof of the following theorem.
\begin{theorem}[Upper bound] \label{thm:upper}
	Let $0\le \ell \le t \le n $. 
	Suppose that $\vx \in \Sigma^n$ and $\vy \in \Sigma^{n + k}$ with $k\ge 0$.
	If  $d_L(\vx, \vy) \geq\ell$, then we~have 
	\begin{equation}\label{eq:main}
		\Delta(\vx,\vy , t, t+k) \leq \frac{\binom{k+2\ell}{\ell}}{(t-\ell)!}n^{t-\ell}\,.
	\end{equation}
\end{theorem}

Next, in Section~\ref{sec:lower}, we provide a matching lower bound. Specifically, we demonstrate the following proposition.

\begin{proposition}[Lower Bound] \label{prop:lower-bound}
	Fix $\ell>0$.
	For $n\ge 4\ell-2$, there exists two sequences $\vX, \vY \in \Sigma^n$ such that $d_L(\vX, \vY) \geq \ell$ and
	\begin{align} 
		\Delta(\vX, \vY, t, t) & \geq  \binom{2\ell}{\ell} D(n-4\ell+2, t-\ell) \label{eq:lower_bound}\\
		& = \frac{\binom{2\ell}{\ell}}{(t-\ell)!}n^{t-\ell} - O(n^{t-\ell-1}), \text{ for all fixed $t\ge \ell$.}\notag
	\end{align}
\end{proposition}

Theorem~\ref{thm-main} now follows from Theorem~\ref{thm:upper} and Proposition~\ref{prop:lower-bound}.

Observe that when we set $\ell=t$ and $k=0$
in \eqref{eq:main}, we have that $N(n,t,t)\le \binom{2t}{t}$. 
On the other hand, for $n\ge 4t-2$, it follows from \eqref{eq:lower_bound} that $N(n,t,t)\ge \binom{2t}{t}$. Therefore, we have determined the exact value of $N(n,t,t)$.

\begin{corollary}\label{cor:lt}
	Set $\ell = t$. For $n\ge 4t-2$, we have that \[N(n,t,t)=\binom{2t}{t} = \Theta(1)\text{ for fixed values of $t$}.\]  
\end{corollary}

Next, we make a conjecture on the {\em exact} value of $N(n,\ell,t)$. In the cases where $\ell\in \{1,2\}$ and $\ell=t$, we remark that the conjecture recovers Theorems~\ref{thm:l1},~\ref{thm:l2} and Corollary~\ref{cor:lt}.

\vspace{2mm}

\noindent{\bf Conjecture}. For $1\le \ell\le t < n$ and sufficiently large $n$, we have  that
\begin{equation}\label{eq:conjecture}
N(n,\ell,t) = N(n-1,\ell,t) + N(n-2,\ell,t-1).
\end{equation} 

Finally, in the spirit of Levenshtein's work \cite{Levenshtein.2001,Levenshtein.2001.jcta}, we provide a polynomial-time reconstruction method in Section~\ref{sec:efficient} for the special case where $\ell=t$.

\begin{proposition}\label{prop:recon}
	Let $\cC$ be an $(t-1)$-deletion-correcting code of length $n$ for some $2 \le t < n$. 
	Suppose further that $\cC$
	has an $(t-1)$-deletion-correcting decoder that runs in $T(n)$ time.
	If we transmit $\vx\in\cC$ over a $t$-deletion channel and obtain $M\triangleq N(n,t,t)+1 = \binom{2t}{t}+1$ distinct outputs, then
	we can determine $\vx$ in time $O(T(n)+Mn)$ time.
\end{proposition}

So, in particular, when $t=2$, the classic Varshamov-Tenengolts (VT) codes are single-deletion-correcting codes equipped with a linear-time decoder \cite{Levenshtein.1966}. Furthermore, \eqref{eq:l2} states that $N(n,2,2)+1=7$  and so, we can uniquely reconstruct a codeword from any distinct seven reads. Then Proposition~\ref{prop:recon} implies that this reconstruction can be done in linear time.

\section{Upper Bound}\label{sec:upper}

In this section, we prove Theorem~\ref{thm:upper} using {\em induction}. 
For convenience, we rewrite the upper bound in the following form that is amenable to an inductive analysis.
\begin{theorem}\label{thm:upper-ind}
Let $\sS(\ell,t,k)$ denotes the statement:
\begin{quote}
	For all $n\ge t$, we have that
	\begin{equation*}
	\Delta(\vx,\vy , t, t+k)
	\begin{cases}
	\leq \frac{\binom{k+2\ell}{\ell}}{(t-\ell)!}n^{t-\ell}, & \mbox{if $t\ge \ell$},\\
	= 0, & \mbox{if $t < \ell$}.
	\end{cases} 
	\end{equation*}
\end{quote}
Then $\sS(\ell,t,k)$ is true for all $\ell,t,k \ge 0$.
\end{theorem}

Note that in Theorem~\ref{thm:upper-ind}, we extend the domain to include the cases where $t < \ell$. 
We justify this in the following lemma where we demonstrate $\sS(\ell,t,k)$ for certain bases cases.

\begin{lemma}[Base Cases] \label{lem:base}
The following instances are true.
\begin{enumerate}[(i)]
	\item $\sS(0,t,k)$ is true for all $t,k\ge 0$.
	\item $\sS(\ell,t,k)$ is true for all $0\le t < \ell$.
\end{enumerate}	
\end{lemma}

\begin{proof}
For (i), we have $\ell=0$. Since $\cD(\vx,\vy , t, t+k)\subseteq \cD_t(\vx)$, 
we apply \eqref{eq:l0} to have $\Delta(\vx,\vy , t, t+k)\le n^t /t!$\,.
For (ii), since $t<\ell$, it follows from the definition of Levenshtein distance that $\cD(\vx,\vy , t, t+k)=\varnothing$.
In other words, $\Delta(\vx,\vy , t, t+k) = 0$.
\end{proof}

Hence, it remains to demonstrates the induction step.
To this end, we define a total order on the set of triples $\{(\ell,t,k) : \ell,t,k\ge 0\}$. 
Specifically, we use $\prec$ to denote the usual {\em lexicographic order} on the triples.
That is, $(\ell,t,k)\prec(\ell_0,t_0,k_0)$ means one of the following:
\begin{itemize}
	\item $\ell < \ell_0$, or
	\item $\ell = \ell_0$ and $t < t_0$, or
	\item $\ell = \ell_0$, $t = t_0$ and $k<k_0$.
\end{itemize}
It is well-known that the lexicographic order defines a total order on the set of triples. Hence, in any nonempty subset of triples, there is always a smallest triple with respect to $\prec$.

Now, we are ready to state the induction step.

\begin{lemma}[Induction Step] \label{lem:induction}
	Suppose that $0<\ell_0\le t_0$ and $k_0\ge 0$. 
	If $\sS(\ell,t,k)$ is true for $(0,0,0) \preceq (\ell,t,k) \prec (l_0,t_0,k_0)$, then $\sS(\ell_0,t_0,k_0)$ is true.
\end{lemma}

As the proof of Lemma~\ref{lem:induction} is fairly technical, we defer the detailed arguments to Subsection~\ref{sec:induction}. 
In what follows, we assume that the lemma is true and complete the induction proof of Theorem~\ref{thm:upper-ind}.
 
\begin{proof}[Proof of Theorem~\ref{thm:upper-ind}]
Suppose otherwise that $\sS(\ell,t,k)$ fails to hold for some triple. 
We choose the smallest such triple $(\ell_0,t_0,k_0)$ with respect to the order $\prec$. 
Since the triple is smallest, we have that $\sS(\ell,t,k)$ is true for all $(\ell,t,k) \prec (l_0,t_0,k_0)$. Furthermore, Lemma~\ref{lem:base} implies that $\ell_0>0$ and $t\ge\ell_0$.

Therefore, the conditions of Lemma~\ref{lem:induction} are met and so, $\sS(\ell_0,t_0,k_0)$ must be true, contradicting our assumption.
\end{proof}

For the rest of this section, we prove the induction step, and we adopt for the following convention. 
For $\vx\in\Sigma^n$, we write $\vx$ as $x_1x_2\cdots x_n$.
In other words, for $1\le i\le n$, the $i$th bit of $\vx$ is denoted by $x_i$.
Furthermore, the length-$i$ prefix of $\vx$ is denoted by $\vx^{(i)}$.
That is, $\vx^{(i)}=x_1x_2\cdots x_i$.
Similarly, for $k\ge 0$, we consider a binary sequence $\vy\in \Sigma^{n+k}$ and let 
$\vy=y_1y_2\cdots y_{n+k}$.

\subsection{Recursion Rules}

Our induction relies on two recursion rules, Lemmas~\ref{lem:recur-dist} and~\ref{lem:recur-ball}. 
To state the recursion rules, we use the following notation used extensively in \cite{Gabrys.2018}.
Given a bit $a\in \Sigma$ and a set $\cX$ of nonempty binary sequences, we use $\cX_a$ to denote the set of sequences in $\cX$ that ends with $a$. Also, we use $\cX \circ a$ to denote the set of sequences obtained by appending $a$ to all sequences in $\cX$. Hence, $|\cX\circ a|=|\cX|$ while $|\cX_a|\le |\cX|$.

Then the following result is folklore.

\begin{lemma} \label{last-bit}
	Given $\vx \in \Sigma^n$ and $a\in\Sigma$, let $i$ be the largest index integer such that $x_i = a$, then we have that $\left(\cD_t(\vx)\right)_a = \cD _{t-(n-i)}\left(\vx^{(i-1)}\right) \circ a$.
	Therefore, the deletion ball centered at $\vx$ can be recursively computed using the rule:
	\[
	\cD_t(\vx) = \left(\cD_t\left(\vx^{(n-1)}\right)\circ x_n\right)  \cup \cD_{t-1}(\vx^{(n-1)})_{\overline{x_n}}
	\]
\end{lemma}

Our first recursion rule provides a lower bound on the Levenshtein distance and is simple modification of the usual recursion rules used in dynamic programming (see for example, \cite{Elzinga.2008}).

\begin{lemma} \label{lem:recur-dist}
	Suppose that $\vx \in \Sigma^n$ and $\vy \in \Sigma^{n + k}$ with $k\ge 0$ and 
	$d_L(\vx, \vy) \geq \ell$.
	\begin{itemize}
	\item When $x_n = y_{n+k}$,
	\begin{equation} \label{eq:dist-last-equal}
		d_L(\vx^{(n-1)}, \vy^{(n+k-1)}) \geq \ell\,.
	\end{equation}
	\item When $x_n \neq y_{n+k}$ and $k=0$,
	\begin{equation} \label{eq:dist-last-diff-k0}
		d_L(\vx^{(n-1)}, \vy) \geq \ell-1,~~~ d_L(\vx, \vy^{(n+k-1)}) \geq \ell-1.
	\end{equation}
	\item When $x_n \neq y_{n+k}$ and $k>0$,
	\begin{equation} \label{eq:dist-last-diff-k1}
		d_L(\vx^{(n-1)}, \vy) \geq \ell-1,~~~~~d_L(\vx, \vy^{(n+k-1)})) \geq \ell.
	\end{equation}
	\end{itemize}
\end{lemma}

\begin{proof}

If $x_{n} = y_{n+k}$, we claim that $\cD(\vx^{(n-1)}, \vy^{(n+k-1)},$ $\ell-1, \ell+k-1) = \varnothing$. Suppose otherwise that $\vz$ belongs to $\cD(\vx^{(n-1)}, \vy^{(n+k-1)}, \ell-1, \ell+k-1)$, then we have $\vz\circ~x_n = \vz \circ y_{n+k} 		\in \cD(\vx, \vy, \ell-1, \ell+k-1)$, which is a contradiction.
\vspace{1mm}

If $x_{n} \ne  y_{n+k}$, we claim  $\cD(\vx^{(n-1)}, \vy, \ell-2, \ell{+}k{-}1){=}\varnothing$. Again, suppose otherwise that $\vz\in \cD(\vx^{(n-1)}, \vy, \ell{-}2, \ell{+}k{-}1)$. Then we have $\vz \in \cD(\vx, \vy, \ell{-}1, \ell{+}k{-}1)$ because $\cD_{\ell-2} (\vx^{(n-1)}) \subseteq \cD_{\ell-1}(\vx)$. This is a contradiction.
\begin{itemize}
\item When $k=0$, we have $\cD(\vx, \vy^{(n+k-1)}, \ell-1, \ell+k-2) = \varnothing$ by symmetry. 
		
\item When $k> 0$, 
we claim $\cD(\vx, \vy^{(n+k-1)}, \ell{-}1, \ell{+}k{-}2) {=} \varnothing$. Again, suppose that $\vz {\in} \cD(\vx, \vy^{(n+k-1)}, \ell{-}1, l{+}k{-}2)$. Then $\vz \in \cD(\vx, \vy, \ell-1, l+k-1)$ because $\cD_{l+k-2}(\vy^{(n+k-1)}) \subseteq \cD_{l+k-1}(\vy)$. This  is a contradiction.
\qedhere
\end{itemize}
	
\end{proof}

The next recursion rule is crucial to our inductive proof.
In particular, we show that we can bound the size of $\cD(\vx,\vy,t,t+k)$ using the corresponding values for the prefixes of $\vx$ and $\vy$.

\begin{lemma} \label{lem:recur-ball}
	Let $\vx \in \Sigma^n$, $\vy \in \Sigma_2^{n+k}$. Then the following are~true.
	
	\begin{itemize}
	\item If $x_n = y_{n+k}$, then
	\begin{align*} 
	\Delta(\vx, \vy, t, t+k) & \leq \Delta(\vx^{(n-1)}, \vy^{(n+k-1)}, t, t+k) \\
	&\hspace{2mm}+ \Delta(\vx^{(n-1)}, \vy^{(n+k-1)}, t-1, t+k-1)\,.
	\end{align*}
	\item If $x_n \neq y_{n+k}$,
	\begin{align*} 
	\Delta(\vx, \vy, t, t+k) & \leq
		\Delta(\vx, \vy^{(n+k-1)}, t, t+k-1) \\
	&\hspace{20mm} + \Delta(\vx^{(n-1)}, \vy, t-1, t+k)\,.
	\end{align*}
	\end{itemize}
\end{lemma}

\begin{proof}In both cases, we apply Lemma~\ref{last-bit}.

When $\vx_n = \vy_{n+k}$, 
\begin{align*}
	&\Delta(\vx, \vy, t, t+k)\\ 
	& = |\cD(\vx, \vy, t, t+k)_{x_n}| 
	+ |\cD(\vx, \vy, t, t+k)_{\overline{x_n}}| \\
	& \le \Delta(\vx^{(n-1)}, \vy^{(n+k-1)}, t, t+k) \\
	&\hspace{1cm}+ \Delta(\vx^{(n-1)}, \vy^{(n+k-1)}, t-1, t+k-1)\,.
\end{align*}

When $\vx_n \neq \vy_{n+k}$, 
\begin{align*}
	&\Delta(\vx, \vy, t, t+k)\\
	& = |\cD(\vx, \vy, t, t+k)_{x_n}| + |\cD(\vx, \vy, t, t+k)_{y_n}| \\
	& \le \Delta(\vx, \vy^{(n+k-1)}, t, t+k-1) + \Delta(\vx^{(n-1)}, \vy, t-1, t+k)\,.
\end{align*}
\end{proof}

\subsection{Proof of Induction Step}
\label{sec:induction}

Finally, we prove Lemma~\ref{lem:induction}.
Specifically, we suppose that $0<\ell_0\le t_0$ and $k_0\ge 0$ and assume $\sS(\ell,t,k)$ is true for $(\ell,t,k) \prec (l_0,t_0,k_0)$. Our aim is to show that $\sS(\ell_0,t_0,k_0)$ is true. In other words, we show that \eqref{eq:main} is true for all $n\ge t_0$ and we do so by induction on $n$.

Suppose that $n = t_0$. 
Then the set $\cD(\vx, \vy, t_0, t_0+k_0)$ is a singleton set that comprises the empty string. So, we have:
	\[
	\Delta(\vx, \vy, t_0, t_0+k_0) = 1 \leq \frac{\binom{2\ell_0}{\ell_0}}{(t_0-\ell_0)!} (t_0+1)^{t_0-\ell_0}.
	\]
	The last inequality holds because $(t_0+1)^{t_0-\ell_0} \geq (t_0-\ell_0)^{t_0-\ell_0} \geq (t_0-\ell_0)!$\, .
	
Next, we assume that \eqref{eq:main} is true for all $n \leq n_0$. 
	We will prove that, for $\vx\in \Sigma^{n_0+1}$ and $\vy\in\Sigma^{n_0+k_0+1}$  with $d_L(\vx,\vy) \geq \ell_0$, we have:
	\[
	\Delta(\vx, \vy, t_0, t_0+k_0) \leq \frac{\binom{2\ell_0}{\ell_0}}{(t_0-\ell_0)!} (n_0+1)^{t_0-\ell_0}\,.
	\]

We have the following two cases.

\vspace{1mm}
\noindent{(I)} When $x_{n_0+1} = y_{n_0+k_0+1}$, it follows from \eqref{eq:dist-last-equal} that $d_L(\vx^{(n_0)}, \vy^{(n_0+ k_0)}) \geq \ell_0$.

When $t_0>\ell_0$, Lemma~\ref{lem:recur-ball} implies that
	\begin{align*}
		&\Delta(\vx, \vy, t_0, t_0+k_0) \\
		& \leq \Delta(\vx^{(n_0)}, \vy^{(n_0+k_0)}, t_0, t_0+k_0) \\
		&\hspace{2cm}+ \Delta(\vx^{(n_0)}, \vy^{(n_0+k_0)}, t_0-1, t_0+k_0-1) \\
		& \leq \frac{\binom{k_0+2\ell_0}{\ell_0}}{(t_0-\ell_0)!} n_0^{t_0-\ell_0} + \frac{\binom{k_0+2\ell_0}{\ell_0}}{(t_0-\ell_0-1)!} n_0^{t_0-\ell_0-1} \\
		& = \frac{\binom{k_0+2\ell_0}{\ell_0}}{(t_0-\ell_0)!} \left(n_0^{t_0-\ell_0} + (t_0-\ell_0)n_0^{t_0-\ell_0-1}\right) \\
		& \leq \frac{\binom{k_0+2l_0}{l_0}}{(t_0-l_0)!} (n_0+1)^{t_0-l_0},
	\end{align*}
	as desired. The last inequality follows from the fact that $(n+1)^{t_0-\ell_0}=\sum_{i=0}^{t_0-\ell_0}\binom{t_0-\ell_0}{i}n^{i}$.

	On the other hand, if $t_0=\ell_0$, we have  $\Delta(\vx^{(n_0)}, \vy^{(n_0+k_0)}, t_0-1, t_0+k_0-1)=0$. It is then not difficult to proceed as above and show that 
	$\Delta(\vx, \vy, t_0, t_0+k_0)\le \binom{k_0+2l_0}{l_0}$.
	
\vspace{1mm}


\noindent{(II)} Suppose that $x_{n_0+1} \ne y_{n_0+k_0+1}$.  When $k_0=0$, \eqref{eq:dist-last-diff-k0} implies that 
	$d_L(\vx, \vy^{(n_0)})\ge \ell_0-1$ and $d_L(\vx^{(n_0)}, \vy) \geq \ell_0-1$.
	Again, we apply Lemma~\ref{lem:recur-ball} to have
	\begin{align*}
		& \Delta(\vx, \vy, t_0+1, t_0+1) \\
		& \leq \Delta(\vx, \vy^{(n_0)}, t_0+1, t_0) + \Delta(\vx^{(n_0)}, \vy, t_0, t_0+1) \\
		& \leq \frac{\binom{2\ell_0-1}{\ell_0}}{(t_0-\ell_0+1)!} n_0^{t_0-\ell_0+1} + \frac{\binom{2\ell_0-1}{\ell_0}}{(t_0-\ell_0+1)!} n_0^{t_0-\ell_0+1} \\
		& =  \frac{2\binom{2\ell_0-1}{\ell_0}}{(t_0-\ell_0+1)!} n_0^{t_0-\ell_0+1} \le \frac{\binom{2\ell_0}{\ell_0}}{(t_0-\ell_0+1)!} (n_0+1)^{t_0-\ell_0+1}.
	\end{align*}
	For the last inequality, observe that $2 \binom{2\ell_0-1}{\ell_0} = \binom{2\ell_0}{\ell_0}$. 

On the other hand, when $k_0>0$, \eqref{eq:dist-last-diff-k1} implies that 
$d_L(\vx, \vy^{(n_0+k_0)})\ge \ell_0$ and 
$d_L(\vx^{(n_0)}, \vy) \geq \ell_0-1$.
Again, applying Lemma~\ref{lem:recur-ball}, we have that
\begin{align*}
	& \Delta(\vx, \vy, t_0, t_0+k_0)\\
	& \leq \Delta(\vx, \vy^{(n_0+k_0)}, t_0, t_0+k_0-1) + \Delta(\vx^{(n_0)}, \vy, t_0-1, t_0+k_0) \\
	& \leq \frac{\binom{k_0+2\ell_0-1}{\ell_0}}{(t_0-\ell_0)!} n_0^{t_0-\ell_0} 
			+ \frac{\binom{k_0+2\ell_0-1}{\ell_0-1}}{(t_0-\ell_0)!} n_0^{t_0-\ell_0} \\
	& = \left( \binom{k_0+2\ell_0-1}{\ell_0} + \binom{k_0+2l_0-1}{\ell_0-1} \right) \frac{n_0^{t_0-l_0}}{(t_0-\ell_0)!} \\
	& \leq \frac{\binom{k_0+2\ell_0}{\ell_0}}{(t_0-\ell_0)!} (n_0+1)^{t_0-\ell_0}\,.
\end{align*}

This completes the induction proof.

\section{Lower Bound}\label{sec:lower}

In this section, we prove Proposition~\ref{prop:lower-bound}. 
Specifically, for $n\ge 4\ell-2$ we explicitly construct two length-$n$ sequences $\vX$ and $\vY$ such that the intersection of their $t$-deletion balls has size at least the quantity defined by \eqref{eq:lower_bound}.

To this end, for $\ell>0$, we consider the following two sequences $\vA_\ell$ and $\vB_\ell$ of length $4\ell-2$.
\begin{align*}
	\vA_{\ell} &\triangleq (1010)^{\ell-1} 10\\
	\vB_{\ell} &\triangleq (0110)^{\ell-1} 01. 
\end{align*}

Here, $\vx^\ell$ denotes the concatenation of $\ell$ copies of $\vx$. 
Now, using $\vA_\ell$ and $\vB_\ell$, we construct the desired sequences $\vX$ and $\vY$.
Specifically, let $\vZ$ be an alternating sequence of length $n-4\ell+2$ that starts with one and we set 
\begin{equation*}
	\vX \triangleq \vA_\ell\vZ \mbox{ and } \vY \triangleq \vB_\ell\vZ.
\end{equation*}

Now, to show that $\vX$ and $\vY$ satisfy the conditions of Proposition~\ref{prop:lower-bound}, 
we first consider $\vA_\ell$ and $\vB_{\ell}$ only, and exhibit the required properties.

Now, we show that $\vA_\ell$ and $\vB_{\ell}$ have the required Levenshtein distance.
To this end, we consider the number of runs in both $\vA_\ell$ and $\vB_{\ell}$.
Formally, for $\vx \in \Sigma^n$, a {\em run} of $\vx$ refers to contiguous repetition of the same bit
and we denote the number of runs in $\vx$ with $R(\vx)$.
The following lemma bounds the changes to $R(\vx)$ when we delete bits from $\vx$.

\begin{lemma} \label{lem:change-run}
	Let $\vx \in \Sigma^n$ and suppose that we delete the bit $x_i$ from $\vx$ to obtain $\vx'$. 
	\begin{itemize}
	\item If $i=1$ or $i=n$, then $R(\vx)-1\le R(\vx')\le R(\vx)$.
	\item If $1<i<n$, then $R(\vx)-2\le R(\vx')\le R(\vx)$.
	\end{itemize}
\end{lemma}

We are ready to show that $\vA_\ell$ and $\vB_\ell$ are far apart in terms of Levenshtein distance.

\begin{lemma} \label{lm:validate_Lev_dis}
	For all $\ell>0$, we have $d_L(\vA_\ell,\vB_\ell)\ge \ell$.
\end{lemma}

\begin{proof}
	First, we count the number of runs in each sequence. Clearly, $R(\vA_{\ell}) = 4\ell-2$ and $R(\vB_{\ell}) = 2\ell$. 
	
	It remains to show that $\cD(\vA_{\ell}, \vB_{\ell}, \ell-1, \ell-1)=\varnothing$.
	Suppose that we delete $\ell-1$ bits from $\vA_{\ell}$ and $\vB_\ell$ to obtain $\vA'$ and $\vB'$, respectively. 
	We consider two cases.
	
	If we delete the first bit of $\vA_\ell$, then by Lemma~\ref{lem:change-run}, 
	\[R(\vA')\ge 4\ell-2 - 1 - 2(\ell-2)=2\ell+1 > R(\vB_\ell)\ge R(\vB').\]
	Since the number of runs of $\vA'$ is strictly greater than that of $\vB'$,  
	we have $\vA'\ne \vB'$.
	
	On the other hand, suppose that we do not delete the first bit of $\vA_\ell$.
	Then, as before, using Lemma~\ref{lem:change-run}, we have that $R(\vA')\ge 2\ell$. However, since the first bit of $\vA_\ell$ and $\vB_\ell$ differ, we need to delete the first bit of $\vB_\ell$ to obtain a common subsequence. 
	Thus, we assume that $R(\vB')< R(\vB_\ell)$. Again, we have that $R(\vA')>R(\vB')$ and so, $\vA'\ne \vB'$.
\end{proof}

Next, we show that the Levenshtein distance of $\vA_\ell$ and $\vB_\ell$ is indeed $\ell$.
Furthermore, the intersection of the corresponding deletion balls is of size at least $\binom{2\ell}{\ell}$.

\begin{lemma} \label{lem:lower-bound}
	For all $\ell>0$, we have that
		$\cD(\vA_{\ell}, \vB_{\ell}, t, t) \geq \binom{2\ell}{\ell}$.
\end{lemma}

\begin{proof}
	We partition the indices $\{1,2,\ldots, 4\ell-2\}$ into pairs $\{1,2\}, \{3,4\},\ldots, \{4\ell-1,4\ell-2\}$
	and split $\vA_\ell$ and $\vB_\ell$ according to these pairs. 
	For convenience, we call a $01$-pair and a $10$-pair on these index pairs {\em a $01$-block} and {\em an $10$-separator}, respectively. 
	Hence, $\vA_\ell$ comprises $2\ell-1$ copies of $10$-separators, while $\vB_{\ell}$ comprises $\ell$ {$01$-blocks}  and $(\ell-1)$ {$10$-separators}. We consider the index set $I=\{1\le i\le 4\ell-2: i\equiv 1,2 \pmod{4}\}$. 
	Notice that the $i$th bits $\vA_\ell$ and $\vB_\ell$ are different if and only if $i$ belongs to $I$.
	Furthermore, $I$ corresponds to the $01$-blocks in $\vB_\ell$ and so, $|I|=2\ell$.
	
	We consider the collection of subsequences of $\vB$ obtained by deleting bits from the $01$-blocks of $\vB_\ell$.
	In other words,
	\[\cE \triangleq \{\ve\in\cD_\ell(\vB_\ell): \text{we delete $\ell$ bits whose indices belong to $I$}\}.\]
	Clearly, $\cE\subset \cD_\ell(\vB_\ell)$. 
	In what follows, we show that $|\cE|=\binom{2\ell}{\ell}$ and $\cE\subset \cD_\ell(\vA_\ell)$.
	
	First, let $J$ and $J'$ be two $\ell$-subsets of $I$. 
	Let $\ve$ and $\ve'$ be the resulting subsequences of $\vB_\ell$ obtained by deleting the indices in $J$ and $J'$, respectively.
	We claim that $\ve\ne \ve'$. Indeed, let us consider the smallest block where $J$ and $J'$ differ.
	That is, $i^*=\min \{1\le i\le \ell: (\{4i-3,4i-2\} \cap J)\ne (\{4i-3,4i-2\} \cap J')\}$.
	Then up to the $(i^*-1)$-th separator, both $\ve$ and $\ve'$ coincides. 
	Since $J$ and $J'$ differ, $\ve$ and $\ve'$ differ for the next two bits and so, $\ve\ne \ve'$ (see Example~\ref{exa-AB}(i)).
	So, there are $\binom{2\ell}{\ell}$ choices for $J$, we have that $|\cE|=\binom{2\ell}{\ell}$.
	
	Next, let $\ve\in \cE$. We claim that $\ve\in \cD_{\ell}(\vA_\ell)$.
	To do so, we insert $\ell$ bits into $\ve$ to obtain $\vA_\ell$.
	Now, recall that only bits in the $\ell$ $01$-blocks of $\vB_\ell$ are deleted. To obtain $\vA_\ell$, we insert zero, one or two bits according to the number of deletions in each $01$-block.
	Specifically, we adopt the following rule.
	\begin{enumerate}[(a)]
		\item If no bits of a $01$-block are deleted, we insert {\em two bits} to obtain two additional $10$-separators.
		\item If only one bit of a $01$-block is deleted, we insert {\em one bit} to obtain an additional $10$-separator.
		\item If the entire $01$-block is deleted, we {\em do not insert any bit}.
	\end{enumerate} 
	
	Now, let the number of $01$-blocks with zero, one, and two deletions be $\alpha, \beta$ and $\gamma$, respectively.
	Because there are $\ell$ $01$-blocks deletions, we have $\alpha + \beta+\gamma=\ell$.
	Also, since the number of deletions is $\ell$, we have that $\beta+2\gamma=\ell$.
	 
	Therefore, the number of $10$-separators created is $2\alpha+\beta=2(\alpha+\beta+\gamma)-(\beta+2\gamma)=\ell$ and so, $\alpha+\beta+\gamma=\ell$.
	Together with the remaining $(\ell-1)$ $10$-separators in $\ve$, we have $(2\ell-1)$ $10$-separators which is $\vA_\ell$ (see Example~\ref{exa-AB}(ii)).
\end{proof}

\begin{example}\label{exa-AB}
Let $\ell=4$. Then 
\begin{align*}
	\vA_4 &= \onezero \onezero \onezero \onezero \onezero \onezero \onezero\,, \\
	\vB_4 &= \zeroone \onezero \zeroone \onezero \zeroone \onezero \zeroone\,.
\end{align*}
Here, we underline the $10$-separators and $I=\{1,2,5,6,9,10,13,14\}$.
\begin{enumerate}[(a)]
\item We choose four out of the eight indices in $I$ to delete from $\vB_4$. Consider $J=\{1,2,5,9\}$ and $J'=\{1,2,5,10\}$. Then the resulting subsequences are 
\begin{align*}
	\ve & = \onezero 1\, \onezero 0\, \onezero 
	01\,,\\
	\ve' & = \onezero 1\, \onezero 1\, \onezero 01.
\end{align*}
Then we follow the steps in the proof of Lemma~\ref{lem:lower-bound} to find that $i^*=3$ 
Indeed, up to the second separator, $\ve$ and $\ve'$ share the same prefix $\onezero 1\, \onezero$ and the next two bits of $\ve$ and $\ve'$ differ.
\item We insert four bits into both $\ve$ and $\ve'$ according to the the proof of Lemma~\ref{lem:lower-bound} and obtain $\vA_4$. Indeed, $\ve,\ve'\in \cD_4(\vA_4)$.
\begin{align*}
\vA_4 &= \onezero 1\ins{0}\,\onezero \ins{1}0\, \onezero \ins{1}0\, 1\ins{0}\, ,\\
&=\onezero 1\ins{0}\,\onezero 1\ins{0}\, \onezero \ins{1}0\, 1\ins{0}\, .
\end{align*}
Here, the inserted bits are highlighted in {\color{blue}blue}.
\end{enumerate}

\end{example}

\begin{proof}[Proof of Proposition~\ref{prop:lower-bound}]
	Since $\vX$ and $\vY$ contains $\vA_\ell$ and $\vB_\ell$, respectively, we have that $d_L(\vX,\vY)\ge d_L(\vA_\ell,\vB_\ell)=\ell$.
	
	Next, we provide a lower bound for $D(\vX,\vY,t,t,)$. 
	Recall the definition of $\cE$ given in the proof of Lemma~\ref{lem:lower-bound} and we consider the following set of length-$(n-t)$ sequences.
	\[\cF = \left\{\ve\vf: \ve\in \cE, \cD_{t-\ell}(\vZ)\right\}.\]
	Since $\cE\subseteq \cD(\vA_\ell,\vB_\ell,\ell,\ell)$, we have that $|\cF|\ge \binom{2\ell}{\ell}|\cD_{t-\ell}(\vZ)|=\binom{2\ell}{\ell}D(n-4\ell+2, t-\ell)$, as required.
\end{proof}

\section{Efficient Reconstruction from Channel Outputs}
\label{sec:efficient}


Throughout this section, let $t\ge 2$ and  $\cC$ be an $(t-1)$-deletion-correcting code of length $n$ equipped with a $(t-1)$-deletion-correcting decoder $\dec$. 
Specifically, $\dec$ is a map from $\Sigma^{n-t+1}$ to $\cC\cup\{\fail\}$ such that 
\begin{equation*}
	\dec (\vy) = 
	\begin{cases}
		\vx,& \mbox{if $\vy\in \cD_{t-1}(\vx)$ and $\vx\in \cC$,}\\
		\fail, & \mbox{if $\vy\notin \bigcup_{\vx\in\cC} \cD_{t-1}(\vx)$. }
	\end{cases}
\end{equation*}

In this section, we propose a simple and efficient algorithm that makes use of $\dec$ to recover a transmitted codeword from their noisy outputs. Specifically, we prove a general version of Proposition~\ref{prop:recon}, 
where we consider a subcode $\cX$ of $\cC$.
\newpage

\begin{proposition}\label{prop:efficient-2}
	Suppose that $\cC$ is a $(t-1)$-deletion-correcting code equipped with a decoder $\dec$ that runs in $T(n)$ time.
	
	Consider a subcode $\cX$ with the property that $D(\vx,\vx',t,t) < M$ for all pairwise distinct $\vx, \vx'\in \cX$.
	
	If we are given $M$ distinct outputs $\vy_1,\vy_2,\ldots, \vy_{M}\in \cD_t(\vx)$, then we are able to reconstruct $\vx$ in $O(T(n)+Mn)$ time. 
\end{proposition}


We provide a high level description of our reconstruction algorithm. First, we pick any two noisy outputs, say, $\vy_1$ and $\vy_2$. Next, we compute two possible candidates $\vx$ and $\vx'$. Finally, we use the remaining $M-2$ outputs to eliminate the incorrect candidate.
The following lemma justifies the first step.

\begin{lemma} \label{lem:rec-from-2}
Let $\vX \in \Sigma^{n}$ and $\vY, \vZ \in \cD_{t}(\vX)$. 
Suppose that $\vU$ is the longest common suffix of $\vY$ and $\vZ$.
Without loss of generality, assume that $\vY = \vB 0 \vU$ and $\vZ = \vC 1 \vU$. 
Then either $\vB 01 \vU \in \cD_{t-1}(\vX)$ or $\vC 10 \vU \in \cD_{t-1}(\vX)$.
\end{lemma}

\begin{proof}
Let $\vX = \vA \vV$ where $\vV$ is the shortest suffix of $\vX$ that contains $\vU$. 
Let $k = |\vV| - |\vU|$. 
Then we have $\cD_{t}(\vX)_{\vU} = \cD_{t-k}(\vA) \circ \vU$. 
In other words, $\vB 0$ and $\vC 1$ are in $\cD_{t-k}(\vA)$, or specifically, $\vB 0 \in \cD_{t-k}(\vA)_{0}$ and $\vC1 \in \cD_{t-k}(\vA)_{1}$. Now there are two cases to consider:
\begin{enumerate}[(i)]
	\item If $\vA = \vA' 0$, then $\cD_{t-k}(\vA)_{1} = \cD_{t-k}(\vA' 0)_{1} = \cD_{t-k-1}(\vA')_{1} \subseteq \cD_{t-k-1}(\vA')$. 
	This means $\vC 1 \in \cD_{t-k-1}(\vA')$.
	Hence, $\vC 10 \in \cD_{t-k-1}(\vA)$.
	Since $\cD_{t-k-1}(\vA) \circ \vU = \cD_{t-1}(\vX)_{\vU}$,
	we have that $\vC 10 \vU \in \cD_{t-1}(\vX)$.
	
	\item Similarly, if $\vA = \vA' 1$, then $\cD_{t-k}(\vA)_{0} = \cD_{t-k}(\vA' 1)_{0} = \cD_{t-k-1}(\vA')_{0} \subseteq \cD_{t-k-1}(\vA')$. 
	This means $\vB 0 \in \cD_{t-k-1}(\vA')$.
	Hence, $\vB 01 \in \cD_{t-k-1}(\vA)$.
	Since $\cD_{t-k-1}(\vA) \circ \vU = \cD_{t-1}(\vX)_{\vU}$,
	we have that $\vB 01 \vU \in \cD_{t-1}(\vX)$. \qedhere
\end{enumerate}
\end{proof}

We are now ready to present our algorithm for Proposition~\ref{prop:efficient-2}.
Recall that $\cX$ is a subcode of a $(t-1)$-deletion-correcting code with the property that $D(\vx,\vx',t,t) < M$ for all pairs of codewords $\vx, \vx'\in \cX$.
\vspace{2mm}

\noindent{\sc Input}: $\vy_1, \vy_2, \ldots, \vy_{M}\in\cD_t(\vx)$ for some $\vx\in\cX$\\
\noindent{\sc Output}: $\vx\in \cX$\\[-3mm]
\begin{enumerate}
\item[(1)] We pick two outputs, $\vy_1, \vy_2$, and set $\vu$ to be the longest common suffix of $\vy_1$ and $\vy_2$. Without loss of generality, assume that $\vy_1 = \va_1 0 \vu$ and $\vy_2 = \va_2 1 \vu$. 
\item[(2)] By Lemma~\ref{lem:rec-from-2}, we have two possible scenarios.
\begin{itemize}
	\item If $\va_1 01 \vu \in \cD_{t-1}(\vx)$, then we can recover $\vx$ by using the $(t-1)$-deletion-correcting decoder. Specifically, we set $\vx_1 \gets \dec (\va_1 01 \vu)$.
	\item Similarly, if $\va_2 10 \vu \in \cD_{t-1}(\vx)$, we can also recover $\vx$ using $\dec$. So, we set $\vx_2 \gets \dec (\vy_2' 10 \vu)$.
\end{itemize}
\item[(3)] Finally, to distinguish between the two scenarios, we use the remaining outputs $\vy_3, \vy_4, \ldots, \vy_M$. Specifically, if $\{\vy_3, \vy_4, \ldots, \vy_M\}\subseteq \cD_t(\vx_1)$, we return the codeword $\vx_1$. Otherwise, we return $\vx_2$.
\end{enumerate}

To complete the proof of Proposition~\ref{prop:efficient-2}, we analyse the running time. 
Clearly, Steps 1 and 2 can be completed in $O(n)$ time and $2T(n)$ time respectively. 
For Step 3, we need to determine if $\vy_j$ is a subsequence of $\vx_i$ for $3\le j\le M$ and $i\in\{1,2\}$.
Since each verification can be completed in $O(n)$ time, Step 3 can be completed in $O(Mn)$ time and the proposition follows.

To conclude this section, we discuss the implication of Proposition~\ref{prop:efficient-2} for the case where $t=2$, that is, the channel that introduces two deletions.

Consider a VT code $\cC$ of length $n$.
As mentioned earlier, $\cC$ is a single-deletion-correcting code that is equipped with a linear-time decoder. 
In \cite{Gabrys.2018} (or Theorem~\ref{thm:l2}), Gabrys and Yaakobi showed that we can uniquely reconstruct any codeword of $\cC$ using seven distinct reads. 
Later, Chrisnata and Kiah considered a subcode $\cX$ of $\cC$ with $2\log_2 n + O(\log_2\log_2 n)$ redundant bits and showed that any codeword of $\cC$ can be reconstructed using five distinct reads \cite{Chrisnata.2021}.
In both cases, naively, we can reconstruct the transmitted word in quadratic time by trying all possibilities for the missing two positions.
However, if we apply the algorithms in Propositions~\ref{prop:recon} and~\ref{prop:efficient-2}, we are able to recover the transmitted word in linear time.

\end{document}